\documentclass[a4paper,11pt]{article}

\usepackage{amsmath}
\usepackage{amssymb}
\usepackage{amsthm}
\usepackage{mathtools}
\usepackage[OT4]{fontenc}
\usepackage[utf8]{inputenc}
\usepackage[margin=3cm]{geometry}
\usepackage{graphicx}
\usepackage{xurl}
\usepackage[numbers,square]{natbib}
\usepackage[unicode,debug,colorlinks=true,linkcolor=black,citecolor=black,urlcolor=black,pdfusetitle=true,breaklinks=true]{hyperref}
\allowdisplaybreaks

\newcommand{\abs}[1]{\left| #1 \right|}
\newcommand{\okra}[1]{\left( #1 \right)}
\newcommand{\kwad}[1]{\left[ #1 \right]}
\newcommand{\klam}[1]{\left\{ #1 \right\}}
\newcommand{\boole}[1]{{\mathbf 1}{\klam{#1}}}
\newcommand{\ceil}[1]{\left\lceil #1 \right\rceil}

\DeclareMathOperator*{\argmax}{arg\, max}

\DeclareMathOperator{\card}{card}

\DeclareMathOperator{\mean}{\mathbf{E}}

\newtheorem{definition}{Definition}

\newtheorem{theorem}{Theorem}
\newtheorem{proposition}{Proposition}

\begin{document}

\title{From Zipf's Law to Neural Scaling through\\ Heaps' Law and
  Hilberg's Hypothesis}

\author{{\L}ukasz D\k{e}bowski%
  \thanks{{\L}ukasz D\k{e}bowski is with the Institute of Computer
    Science, Polish Academy of Sciences, ul. Jana Kazimierza 5, 01-248
    Warszawa, Poland, e-mail: ldebowsk@ipipan.waw.pl.}%
  %\thanks{Manuscript received ...; revised ...}%
}

\date{}

\begin{titlepage}

\maketitle

\begin{abstract}
  We inspect the deductive connection between the neural scaling law
  and Zipf's law---two statements discussed in machine learning and
  quantitative linguistics. The neural scaling law describes how the
  cross entropy rate of a foundation model---such as a large language
  model---changes with respect to the amount of training tokens,
  parameters, and compute. By contrast, Zipf's law posits that the
  distribution of tokens exhibits a power law tail.  Whereas similar
  claims have been made in more specific settings, we show that the
  neural scaling law is a consequence of Zipf's law under certain
  broad assumptions that we reveal systematically.  The derivation
  steps are as follows: We derive Heaps' law on the vocabulary growth
  from Zipf's law, Hilberg's hypothesis on the entropy scaling from
  Heaps' law, and the neural scaling from Hilberg's hypothesis. We
  illustrate these inference steps by a toy example of the Santa Fe
  process that satisfies all four statistical laws.
  \\[1em]
  \textbf{Key words}: neural scaling law, Zipf's law, Heaps' law,
  Hilberg's hypothesis, Santa Fe processes
  \\[1em]
  \textbf{MSC 2020:} 94A17, 60G10, 91F20, 68T07
\end{abstract}
%MSC 2000: 
% 60G10 - Stationary processes
% 94A17 - Measures of information, entropy
% 91F20 - Linguistics
% 68T07 - Artificial neural networks and deep learning
% 68T50 - Natural language processing
\end{titlepage}

\maketitle

\section{Introduction}
\label{secIntroduction}

It has been increasingly recognized in the machine learning literature
\citep{Hutter21, MaloneyRobertsSully22, MichaudOther23,
  CabannesDohmatobBietti24, NeumannGros25, PanOther25} that the
\emph{neural scaling law} observed for contemporary foundation
models---such as large language models---may arise as a consequence of
\emph{Zipf's law} or similar distributional regularities of natural
language. This emerging perspective suggests that some remarkable
empirical regularities in large-scale deep learning need not be
explained solely by architectural, optimization, or hardware
considerations, but instead reflect intrinsic statistical constraints
of \emph{natural texts}---human linguistic data.

In the present paper, we aim to provide an actually simple but
systematic derivation of this chain of implications that extends
earlier results in information theory, probability theory, and
quantitative linguistics. We will prove formal theorems about as
general stochastic processes as possible rather than experiment with
particular empirical data. Our goal is to consolidate knowledge of
several scientific disciplines by means of mathematical deduction. We
try to avoid a too abstract formalism to make this paper more
accessible.

The conceptual trajectory is as follows. Beginning with Zipf's law for
word frequencies, we derive \emph{Heaps' law} for vocabulary
growth. From Heaps' law we extract \emph{Hilberg's hypothesis} on the
sublinear growth of block entropy.  Finally, we show that Hilberg's
hypothesis leads to the \emph{neural scaling} that ties model
performance to the amounts of training data, model parameters, and
training compute. In parallel, we discuss \emph{Santa Fe
  processes}---toy stochastic sources that exhibit these statistical
laws.  Each link in this chain introduces its own assumptions. By
isolating these steps, we hope to illuminate where the derivations
might be strengthened in the future.

\paragraph{Further organization of the article.}

% The organization of this paper is as follows. %
In Section~\ref{secPreliminaries}, we map preliminaries and the
research context, including basic notation, power laws at large, and
prior theoretical connections. %
Section~\ref{secResults} contains the formal statement of the main
theorems, preceded by a discussion of their novelty and followed by a
consideration of their relevance. %
Section~\ref{secConclusion} concludes the article. %
The paper is followed by two appendices. %
Appendix \ref{secBitsTypes} contains general results about counting
types and bits of information for diverse stochastic sources.  Proofs
of the main theorems can be found in Appendix \ref{secImplications}.

\section{Preliminaries and context}
\label{secPreliminaries}

This section begins with an explanation of mathematical notation, to
be followed by a discussion of empirical power laws in quantitative
linguistics and machine learning, to be concluded by the attempts of
relating them formally.

\subsection{Basic notations}
\label{secNotations}

Consider a countable alphabet $\mathbb{X}$, which need not be
finite. All considered random variables live on the measurable space
of doubly infinite sequences $(\Omega,\mathcal{J})$, where
$\Omega=\mathbb{X}^{\mathbb{Z}}$ and $\mathcal{J}\subset 2^{\Omega}$
is the smallest $\sigma$-field generated by cylinder sets. Function
$P:\mathcal{J}\to[0,1]$ is the ``true'' probability measure that need
not have a finite description or be computable.

\paragraph{Tokens and types.}
Notation $(X_t)_{t\in\mathbb{Z}}$ denotes an integer-\-time stochastic
process consisting of countable-\-alphabet random variables
$X_j:\Omega\ni(x_t)_{t\in\mathbb{Z}}\mapsto x_j\in\mathbb{X}$.  As in
quantitative linguistic research, we distinguish tokens and types.
That is, string
\begin{align}
  X_j^k:=(X_j,X_{j+1},\ldots,X_k)
\end{align}
consists of random elements called tokens, whereas calligraphically
written set
\begin{align}
  \mathcal{X}_j^k:=\klam{X_j,X_{j+1},\ldots,X_k}
\end{align}
consists of random elements called types (i.e., unique tokens). We
adopt that $X_1^0$ is the empty string and $\mathcal{X}_1^0$ is the
empty set.  Analogous notations are applied to processes denoted by
other letters such as $(K_t)_{t\in\mathbb{Z}}$ or
$(Z_k)_{k\in\mathbb{N}}$.

\paragraph{Information measures.}
Expectation with respect to the ``true'' probability measure $P$ is
denoted $\mean X:=\int X dP$.  Random variable $P(X)$ for a random
variable $X$ is defined as
\begin{align}
  P(X)(\omega):=P(X=x)\iff X(\omega)=x.
\end{align}
This convention allows us to shorten formulas by writing $P(X)$
instead of $P(X=x)$.  In particular, notation
\begin{align}
  H(X):=\mean(-\log P(X))
\end{align}
denotes the Shannon entropy of a simple random
variable $X$. The conditional entropy is $H(X|Z):=H(X,Z)-H(Z)$, the
mutual information is $I(X;Y):=H(X)-H(X|Y)$, and the conditional
mutual information is $I(X;Y|Z):=H(X|Z)-H(X|Y,Z)$.

\paragraph{Extension to $\sigma$-fields.}
For random variables that are not simple, i.e., assume an infinite
number of values, we use a convenient generalization of the above
Shannon information measures to arbitrary $\sigma$-fields. For a
textbook treatment of their calculus, see \citep[Sections 5.3 and
5.4]{Debowski21}. The original contributions were made by
\citet{GelfandKolmogorovYaglom56en}, \citet{Dobrushin59en},
\citet{Pinsker60en}, \citet{Wyner78}, and \citet{Debowski09,
  Debowski20}.

\paragraph{Random measures.}
Symbol $Q$ most often denotes a random probability measure on the
measurable space $(\Omega,\mathcal{J})$, which is ``learned'' from a
finite sample such as $X_j^k$, i.e., $Q=f(X_j^k)$ for a certain
(stochastic) function $f$. Mind that the entropy of $Q$ equals
\begin{align}
  H(Q)=-\sum_q P(Q=q)\log P(Q=q),
\end{align}
where $q$ are distinct values (concrete measures) that $Q$ assumes. In
particular, it is not the entropy rate of process
$(X_t)_{t\in\mathbb{Z}}$ with respect to $Q$. In contrast to $P$,
values of measure $Q$ are often computable objects, expressed via a
finite number of finite-resolution parameters. In particular,
$H(Q)\le H(X_j^k)$ for $Q=f(X_j^k)$ and a deterministic function $f$.

\paragraph{Expected cardinalities.}
An important analogue of Shannon entropy in our considerations is the
expected cardinality of a random finite set $\mathcal{X}$ denoted
\begin{align}
  V(\mathcal{X}):=\mean\card\mathcal{X}.
\end{align}
Respectively, the analogue of conditional entropy $H(X|Z)=H(X,Z)-H(Z)$
is the expected cardinality of set difference
$V(\mathcal{X}\setminus\mathcal{Z})
=V(\mathcal{X}\cup\mathcal{Z})-V(\mathcal{Z})$ and the counterpart of
mutual information $I(X;Y)=H(X)+H(Y)-H(X,Y)$ is the expected
cardinality of intersection
$V(\mathcal{X}\cap \mathcal{Y}) = V(\mathcal{X}) + V(\mathcal{Y}) -
V(\mathcal{X}\cup \mathcal{Y})$, cf.\ the idea of $I$-measure
\citep{Hu62, Yeung02}. However, this is an incomplete analogy since in
general there is no random variable $Z=f(X,Y)$ such that
$H(Z)=I(X;Y)$, cf.\ \citep{GacsKorner73,Wyner75}.

\subsection{Power laws at large}
\label{secPowerLaws}

Let us briefly present four power laws to be related in our formal
reasoning---in the chronological order of their discovery.

\paragraph{Zipf's law.}
The oldest known of quantitative linguistic laws, Zipf's law asserts
that, for texts in natural language, the frequency of the $k$-th most
common type of word decays approximately as $k^{-\alpha}$, where
$\alpha\approx 1$. This regularity was noticed over one century ago
\citep{Estoup16, Condon28, Zipf35}.  An empirical study of this law
across one hundred languages can be found in \citep{MehriJamaati17}.
Similar power-law distributions are observed also in ecology,
sociology, economics, and physics \citep{Zipf49}, being a hallmark of
complex systems. The literature on Zipf's law is scattered over
diverse venues. Historical references have been compiled by
\citet{LiWWW}. 

It has been known that Zipf's law is a stylized fact rather than
follows simple exact formulas.  Despite these complications, we will
assume that stationary processes $(K_t)_{t\in\mathbb{Z}}$ over natural
numbers with the approximate Zipf probability distribution
\begin{align}
  \label{Zipf}
  A_Z^- k^{-\alpha_Z}\le P(K_t=k)&\le A_Z^+ k^{-\alpha_Z},
  & \alpha_Z&>1, & A_Z^\pm&>0,
\end{align}
provide a compact description of heavy-tailed data
distributions. Statistical law (\ref{Zipf}), called also the
Zipf-Mandelbrot law after \citep{Mandelbrot54}, will be our departure
point.

\paragraph{Heaps' law.}
A closely related law, Heaps' law, also known as Herdan's law,
describes the vocabulary growth of natural texts
\citep{KuraszkiewiczLukaszewicz51en, Guiraud54, Herdan64,
  Heaps78}. Heaps' law is also studied for large language models
\citep{TaoOther24, LaiOther23}.  According to this law, the expected
number of unique types in the first $t$ tokens of the process
$(K_t)_{t\in\mathbb{Z}}$ considered in (\ref{Zipf}) increases like a
sublinear power-law function,
\begin{align}
  \label{Heaps}
  A_V^- t^{\beta_V}\le V(\mathcal{K}_1^t)&\le A_V^+ t^{\beta_V},
  & 0<\beta_V&<1, & A_V^\pm&>0.
\end{align}
Heaps' law (\ref{Heaps}) is widely viewed as a direct consequence of
Zipf's law (\ref{Zipf}) with $\beta_V=1/\alpha_Z$ but there are some
complications, which we will inspect further. A naive estimation often
yields $\alpha_Z\approx 1$ and $\beta_V\approx 0.8$ for million-token
corpora (novel-sized texts). There are also problems with languages
that use the ideographic script \citep{TanakaIshii21}.

\paragraph{Hilberg's hypothesis.}
A reinterpretation of Shannon's early findings from 1950's
\citep{Shannon48, Shannon51}, Hilberg's hypothesis, also called
Hilberg's law, was developed around 1990's, mostly in the physics of
complex systems \citep{Hilberg90, EbelingNicolis92,
  BialekNemenmanTishby01b, CrutchfieldFeldman03, Debowski06,
  Debowski11b}.  According to this hypothesis, for a stationary
process $(X_t)_{t\in\mathbb{Z}}$ that models text in natural language,
the block entropy of the first $t$ word tokens contains a sublinear
power-law component,
\begin{align}
  \label{Hilberg}
  A_H^- t^{\beta_H}\le H(X_1^t)-ht&\le A_H^+ t^{\beta_H},
  & 0<\beta_H&<1, & A_H^\pm&>0.
\end{align}
where we apply the entropy rate
\begin{align}
  \label{EntropyRateInf}
  h:=\inf_{t\in\mathbb{N}} \frac{H(X_1^t)}{t}.
\end{align}

The original rough estimate by \citet{Hilberg90} was $\beta_H=0.5$ and
was based on a meager evidence for $t\le 100$. If we estimate the
entropy by the prediction-by-partial matching (PPM) algorithm
\citep{Ryabko84en2, Ryabko88en2, ClearyWitten84}, then law
(\ref{Hilberg}) holds more universally and uniformly than Zipf's or
Heaps' laws. The empirical estimate $\beta_H\approx 0.8$ obtained for
the PPM algorithm is quite stable for billion-token corpora and does
not differ significantly across typologically diverse languages, using
either alphabetic or ideographic scripts
\citep{TakahiraOther16,TanakaIshii21}. Thus Hilberg's law seems a
plausible candidate for a statistical language universal
\citep{TanakaIshii21}.

\paragraph{Neural scaling.}
The universal power-law behavior of information measures, when applied
to natural big data, can be further confirmed by experiments with
large foundation models---language or multimodal models in particular,
developed in the beginning of 2020's \citep{DevlinOther19,
  RadfordOther19, BrownOther20, ThoppilanOther22,
  ChowderryOther22}. These advanced statistical models build upon deep
neural networks \citep{BengioLeCunHinton21}, word embeddings
\citep{MikolovOther13}, and the transformer architecture
\citep{VaswaniOther17}. Foundation models can be regarded as a
game-changer in the research of language and cognition, as they allow
to test probabilistic hypotheses about human language on an
unprecedented scale and detail \citep{FutrellMahowald25,
  FutrellHahn25}.

A particularly salient empirical finding is that the predictive
performance of these models improves with scale in a power-law
fashion. The neural scaling law characterizes how the loss function of
a foundation model---typically measured by the cross entropy rate on a
test dataset---decreases as training data $t$, model size $n$, and
compute $c$ increase \citep{HestnessOther17, KaplanOther20,
  HenighanOther2020, HernandezOther21, HoffmannOther22, PearceSong24,
  LiOther25}.  Simplifying particular empirical observations and
ignoring complex interactions among $t$, $n$, and $c$, this
power-law relationship can be approximated as
\begin{align}
  \label{NeuralT}
  A_T^- t^{-\gamma_T}\le
  h(s,t,\infty,\infty)-h
  &\le A_T^+ t^{-\gamma_T},
    & 0<\gamma_T&<1, & A_T^\pm&>0,
  \\
  \label{NeuralN}
  A_N^- n^{-\gamma_N}\le
  h(s,\infty,n,\infty)-h
  &\le A_N^+ n^{-\gamma_N},
    & 0<\gamma_N&<1, & A_N^\pm&>0,
  \\
  \label{NeuralC}
  A_C^- n^{-\gamma_C}\le
  h(s,\infty,\infty,c)-h
  &\le A_C^+ c^{-\gamma_C},
    & 0<\gamma_C&<1, & A_C^\pm&>0,
\end{align}
for a fixed $s<\infty$, where $h$ is the entropy rate and $h(s,t,n,c)$
is the cross entropy rate of a foundation model trained on $t$ tokens
and tested on $s$ tokens with the amount of parameters $n$ and the
amount of compute $c$. Laws (\ref{NeuralT})--(\ref{NeuralC}) hold for
trillion-token corpora.

\subsection{Prior connections and explanations}
\label{secPrior}

The history of theoretical explanations of power laws is as ancient as
empirical observations of these regularities \citep{LiWWW}. To explain
statistical regularities observed in natural language is a
long-standing goal of quantitative linguistics.
% In Section \ref{secPowerLaws}, we have already mentioned the most
% distinct kinds of explanations of Zipf's law such as monkey-typing
% \citep{Mandelbrot54, Miller57}, preferential attachment
% \citep{Simon55}, also known as the Chinese restaurant process
% \citep[page 92]{Aldous85}, and game theory
% \citep{HarremoesTopsoe05}.
We envision that a similarly systematic approach to Hilberg's law and
neural scaling may succeed as well. A desired unified theory of
language and large language models should predict the functional forms
of all these laws simultaneously and predict the values of their
parameters. Let us comment briefly on the progress made so far.

\paragraph{Quantitative linguistics.}
In spite of sheer cross-disciplinary literature coverage (or maybe
because of that), there is no single explanation of Zipf's law.  The
law can be explained by diverse mechanisms ranging from monkey-typing
\citep{Mandelbrot54, Miller57}, through preferential attachment
\citep{Simon55}, also known as the Chinese restaurant process
\citep[page 92]{Aldous85}, to potential links with game theory
\citep{HarremoesTopsoe05}, semantics, and information theory
\citep{Debowski21}.

Explaining Zipf's and Heaps' laws is not so straighforward because of
multiple phenomena that break the simple picture of power laws in
natural language. These challenges include two-regime rank-frequency
plots \citep{FerrerSole01b, MontemurroZanette02, PetersenOther12,
  GerlachAltmann13, WilliamsOther15, ChacomaZanette20}, log-log
convexity of the vocabulary growth \citep{FontClosCorral15}, monotone
decreasing or $U$-shaped hapax rates \citep{Fan10, Debowski25}, and
finite-size effects \citep{LuZhangZhou10}.  One should also note that
expected values do not explain all observed phenomena since variances
are large for language data, due to Taylor's law
\citep{KobayashiTanakaIshii18, TanakaIshii21}. One also observes
burstiness of words, i.e., departure of recurrence times from the
Poisson process \citep{AltmannPierrehumbertMotter09} and long-range
correlations, i.e., slow decay of mutual information between two word
tokens separated by a given lag \citep{WieczynskiDebowski25}.

In spite of these challenges, it is fascinating that the empirical
marginal distribution of words, summarized in Zipf's and Heaps' laws,
can be quite well modeled by non-parametric and parametric urn
models---IID sources of words \citep{Baayen01, Milicka09, Milicka13,
  Davis18, Debowski25}.  For this reason, it is sound to investigate
Zipf's law and Heaps' law mathematically as if the generating process
were a memoryless source \citep{Karlin67, Khmaladze88, PitmanYor97,
  Baayen01} or more generally a stationary process---as it will be
assumed in this paper. Such a theory is capable of deriving Heaps' law
(\ref{Heaps}) from Zipf's law (\ref{Zipf}) under broad conditions.

\paragraph{Santa Fe processes.}  Since Hilberg's law (\ref{Hilberg})
does not hold for IID and finite-state sources
\citep{CrutchfieldFeldman03, Debowski21b}, this observation might be
considered a witness to large memory and complex structure of natural
texts.  This view is somewhat inaccurate, however. Large memory does
not necessitate complex structure in an intuitive sense.  A simple
stochastic source that satisfies condition (\ref{Hilberg}) is the
Santa Fe process described by \citet{Debowski05en, Debowski09,
  Debowski11b} and later rediscovered by \citet{Hutter21} in the
context of machine learning.

The idea of the Santa Fe process $(X_t)_{t\in\mathbb{Z}}$ is to
decompose each token $X_t$ as a pair of a natural number $K_t$ and an
additional bit---which is copied from a certain sequence of bits
$(Z_k)_{k\in\mathbb{N}}$ by taking the bit at position $K_t$. Thus,
each text token $X_t$ may be written as a pair
\begin{align}
  \label{SantaFe}
  X_t=(K_t,Z_{K_t}),
\end{align}
where $(Z_k)_{k\in\mathbb{N}}$ is the sequence of bits, called the
knowledge, and $(K_t)_{t\in\mathbb{Z}}$ is the sequence of natural
numbers, called the narration. To draw attention to a important
detail, the extra bit $Z_k$ has the property of making token $(k,Z_k)$
one bit more unpredictable the first time it is observed and its
effect vanishes later.

The terms ``knowledge'' and ``narration'' were chosen because of a
semantic interpretation of Santa Fe processes, discussed by
\citet{Debowski11b, Debowski21, Debowski21b}.  We may interpret that
pairs $(K_t,Z_{K_t})$ are statements that describe bits of sequence
$(Z_k)_{k\in\mathbb{N}}$ in an arbitrary order but in a
non-contradictory way. Namely, if statements $(k,Z_k)$ and
$(k',Z_{k'})$ describe the same bit $(k=k')$ then they assert the same
value $(Z_k=Z_{k'})$. Thus, we may interpret that sequence
$(Z_k)_{k\in\mathbb{N}}$ expresses some unbounded immutable general
knowledge which is only partially accessed and described in finite
texts.

Contrary to intuitions about complex structures, Hilberg's law arises
also in the highly simplified setting of Santa Fe processes. In
particular, it suffices to assume that narration
$(K_t)_{t\in\mathbb{Z}}$ is an IID source with the Zipf distribution
(\ref{Zipf}) and knowledge $(Z_k)_{k\in\mathbb{N}}$ is a sequence of
independent fair coin flips, independent of narration
$(K_t)_{t\in\mathbb{Z}}$. Under these conditions, process
$(X_t)_{t\in\mathbb{Z}}$ is exchangeable and we obtain Heaps' law
(\ref{Heaps}) and Hilberg's law (\ref{Hilberg}) with
$\beta_V=\beta_H=1/\alpha_Z$, similar results having been obtained by
\citet{Debowski11b,Debowski21}.

\paragraph{Machine learning.}
% We also notice the relevant more recent activity concerning neural
% scaling done in the field of machine learning.
There is a rapidly growing body of literature that seeks to explain
the neural scaling law and similar quantitative observations
concerning foundation models \citep{BelkinOther19, LouartOther18,
  BartlettOther20, BahriOther21, ZhangOther21, Belkin21, Hutter21,
  MaloneyRobertsSully22, RobertsOther22, WeiOther22, SharmaKaplan22,
  MichaudOther23, AchilleSoatto26, BarkeshliAlfaranoGromov26,
  CagnettaOther26}. Many of these works involve sophisticated
mathematical frameworks, including applications of random matrix
theory and techniques from theoretical physics such as Feynman
diagrams.

From this perspective, we think that it is important to appreciate
works which derive neural scaling from facts as simple as Zipf's law
\citep{Hutter21, MaloneyRobertsSully22, MichaudOther23,
  CabannesDohmatobBietti24, NeumannGros25, PanOther25}. In the
following, we compare our results with several most related works in
this area, namely \citep{Hutter21, MichaudOther23, AchilleSoatto26,
  CagnettaOther26}.

Chronologically, in the first work, \citet{Hutter21} derived the
neural scaling for sample length (\ref{NeuralT}) for the specific
model of a Santa Fe process (\ref{SantaFe}) with the IID narration
$(K_t)_{t\in\mathbb{Z}}$ distributed according to the exact
Zipf-Mandelbrot law $P(K_t=k)\propto k^{-\alpha_Z}$ and the
fair-coin-flip knowledge $(Z_k)_{k\in\mathbb{N}}$. Staying in the IID
regime, \citet{Hutter21} was also able to develop laws not only for
expectations but also variances.

In the second work, \citet{MichaudOther23} considered the same
specific Santa Fe process as considered by
\citet{Hutter21}. Independently, they assigned it a more semantic
interpretation like in \citet{Debowski11b, Debowski21, Debowski21b},
and derived the neural scaling laws (\ref{NeuralT})--(\ref{NeuralN})
with exponents
\begin{align}
  \label{PredictedExponents}
  \gamma_T &= \frac{\alpha_Z - 1}{\alpha_Z}=1-\beta_V,
  &
    \gamma_N &= \alpha_Z-1=\frac{1-\beta_V}{\beta_V}.
\end{align}
for the Heaps law exponent $\beta=1/\alpha_Z$.

Identifying the Heaps law exponent $\beta_V$ with the Hilberg law
exponent $\beta_H$, for estimate $\beta_H\approx 0.8$ reported by
\citet{TakahiraOther16}, formulas (\ref{PredictedExponents}) yield
$\gamma_T= 0.2$ and $\gamma_N= 0.25$. This is quite a difference to
the values $\gamma_T\approx 0.095$ and $\gamma_N\approx 0.076$
reported by \citet{KaplanOther20}. We note that these estimates were
obtained for corpora of a different magnitude---billions of tokens in
the case of \citet{TakahiraOther16} to be contrasted with trillions of
tokens in the case of \citet{KaplanOther20}.

Formulas (\ref{PredictedExponents}) imply inequality
$\gamma_T<\gamma_N$, which is not confirmed by \citet{KaplanOther20}
but corroborated by later experiments \citep{HenighanOther2020,
  HernandezOther21, HoffmannOther22, PearceSong24, LiOther25}.  Figure
15 by \citet{MichaudOther23} depicts a huge variation in the empirical
estimates of parameters $\gamma_T$ and $\gamma_N$ across different
studies, depending on the cited experiment or even within the
particular experimental paper.  Thus theoretical derivations of the
neural scaling law probably cannot fully explain the diversity of
empirically obtained parameter values.

Moreover, a theoretical derivation of the neural scaling law for
compute (\ref{NeuralC}) seems to have not been achieved yet.  From
this perspective, it may be important to mention the work of
\citet{AchilleSoatto26}, who sought to link Hilberg's law with
time-bounded Kolmogorov complexity \citep{Levin73, Levin84}. This work
does not concern neural scaling (\ref{NeuralT})--(\ref{NeuralC}) but
rather intelligent agents that find themselves under pressure to
memorize patterns if they are rewarded for saving time. Still, it may
be inspiring for future derivations of the compute law
(\ref{NeuralC}).

In the fourth, more recent work, \citet{CagnettaOther26} developed a
neural scaling formula from Hilberg's law, which they have
hypothesized without any explicit reference to the earlier work in
this area. Their reasoning is different than one to be presented in
this paper. In particular, \citet{CagnettaOther26} assume additionally
a power-law decay of token-token correlation strength, which we also
found empirically in language data
\citep{WieczynskiDebowski25}. However, this condition is not necessary
in our theoretical derivation of neural scaling from Hilberg's
law. Our assumption (\ref{ConditionalBound}) is the closest to this
power-law decay of token-token correlation out of all our
formulas. This assumption, however, is only used to derive Heaps' law
from Zipf's for long-range dependent data.

\medskip%
As we can see, there is a lot of partial insight in the machine
learning literature. Besides a succinct review of references to much
earlier works in quantitative linguistics, matematics, information
theory, and complex systems, which we tried to sketch, a clear
mathematical message is due.

\section{Statement of results}
\label{secResults}

In this section, we present the core of our theoretical results. We
begin with explaining what is a novel development, then we state the
theorems formally, and we end up with discussing the relevance of
assumptions.

\subsection{Novel ideas}
\label{secNovelty}

As we have mentioned in the first sentence of this article, it is
increasingly more frequently acknowledged that the neural scaling law
observed for foundation models may be caused by the inherent power
laws of natural language.  Whereas there are some formal results in
this in vein discussed in Section \ref{secPrior}, the present paper
develops and strengthens the ideas from an earlier unpublished attempt
to derive neural scaling from Hilberg's hypothesis
\citep{Debowski23}. The present paper supplies a simple-minded
baseline that, in contrast to \citep{Debowski23}, takes into account
also the effect of limiting the amount of training compute. Particular
Theorems \ref{theoSandwichZipf}--\ref{theoHilbergNeural} to be
discussed in Section \ref{secTheorems} contribute to an effort of
understanding large-scale statistical regularities of language and
language models within a common theoretical framework.  Their novelty
lies in five points listed below.

\paragraph{A hierarchy of statistical laws.}
The principal conceptual contribution of this paper is the derivation
of the chain of implications
\begin{align}
  \label{Implications}
  \text{Zipf}
  \xRightarrow{\text{Theorems 1--3}} \text{Heaps}
  \xRightarrow{\text{Theorem 4}} \text{Hilberg}
  \xRightarrow{\text{Theorem 5}} \text{neural scaling},
\end{align}
which link four regularities that are usually investigated in
different research communities. In particular, Zipf's and Heaps' laws
originated in quantitative linguistics, Hilberg's law arose in
information-theoretic studies of long-range dependence in complex
systems, whereas neural scaling laws emerged from empirical
investigations in machine learning. Our results suggest that these
phenomena need not be viewed as independent observations but may
instead form a chain of progressively more general implications.

\paragraph{Increasing generality of assumptions.}
Indeed, the individual propositions are established under assumptions
of different strength. Theorems \ref{theoUpperHeaps} and
\ref{theoLowerHeaps} concern IID and some pseudo-mixing processes,
respectively.  Theorem \ref{theoHilbergSantaFe} is proved for
stationary Santa Fe processes, which provide a mathematically
tractable model exhibiting both Zipfian and Hilberg-type behavior.
Finally, Theorem \ref{theoHilbergNeural} is formulated in an
information-theoretic framework that can accommodate also certain
non-stationary sources. Although no single theorem covers all settings
of practical interest, the sequence of results demonstrates that more
global statistical laws can be deduced from increasingly weaker
assumptions.

\paragraph{Differential forms of Heaps' and Hilberg's laws.}
A conceptual novelty of this paper is the use of differential versions
of Heaps' and Hilberg's laws. Classical formulations describe the
growth of cumulative quantities such as vocabulary size or excess
entropy. Deriving neural scaling, however, requires controlling the
marginal gains obtained from additional data. This motivates the
introduction of differential laws that quantify incremental growth
rates. These laws constitute the natural intermediate quantities
linking linguistic statistics to learning curves.

\paragraph{Implications for neural scaling.}
Theorem \ref{theoHilbergNeural} suggests that neural scaling may be
understood as a consequence of the statistical structure of the data,
rather than solely as a property of particular optimization
algorithms, architectures, or parameterizations. The theorem does not
predict the exact scaling exponents observed in practice. Instead, it
provides information-theoretic lower bounds on the achievable rates of
improvement with increasing data, model complexity, and computational
resources. In this sense, the result may be viewed as a structural
constraint on learning curves that holds independently of many
implementation details.

\paragraph{Information-theoretic modeling of resources.}
The assumptions concerning model size (\ref{ParameterBound}) and
training compute (\ref{ComputeBound}) are stated in terms of entropy
bounds rather than specific computational architectures. This choice
reflects the idea that any resource constraint ultimately limits the
amount of information that can be stored, processed, or inferred. By
expressing the assumptions at this level of abstraction, the resulting
bounds apply simultaneously to a wide range of learning systems.

\subsection{Formal propositions}
\label{secTheorems}

This paper strives to identify mathematical implications in the most
general setting and to make interdisciplinary connections. The main
contributions of the present paper are organized as the chain of
derivations (\ref{Implications}).  The proofs of the respective
Theorems \ref{theoSandwichZipf}--\ref{theoHilbergNeural} can be found
in Appendix \ref{secImplications}, whereas the discussion of the
strength of these formal propositions is deferred to Section
\ref{secRelevance}.

Now let us announce these results formally.  First of all, Zipf's law
with the Mandelbrot correction (\ref{Zipf}) implies structural
inequalities (\ref{UpperStructure}) and (\ref{LowerStructure}) that
matter in further developments.
\begin{theorem}
  \label{theoSandwichZipf}
  Let a sequence $(p_k)_{k\in\mathbb{N}}$, where $p_k\ge 0$ and
  $\sum_{k\in\mathbb{N}} p_k=1$. Suppose that we have Zipf's law with
  the Mandelbrot correction
  \begin{align}
    \label{SandwichZipf}
    A_Z^-k^{-1/\beta}\le p_k\le A_Z^+k^{-1/\beta}
  \end{align}
  for some constants $\beta\in(0,1)$, $A_Z^\pm>0$, and all
  $k\in\mathbb{N}$. Then for
  \begin{align}
    A_S^+&:=\okra{\frac{1}{A_Z^+}}^{1-\beta},
    &
    A_S^-&:=\frac{A_Z^-}{1/\beta-1}\okra{\frac{1}{A_Z^+}}^{1-\beta},
  \end{align}
  and all $p>0$, we have two structural inequalities
  \begin{align}
    \label{UpperStructure}
    \sum_{k\in\mathbb{N}} \boole{p_k\ge p}&\le A_S^+p^{-\beta},
    \\
    \label{LowerStructure}
    \sum_{k\in\mathbb{N}} p_k\boole{p_k\le p}&\ge A_S^-p^{1-\beta}.
  \end{align}
\end{theorem}
Bounds similar to (\ref{UpperStructure}) and (\ref{LowerStructure})
were discussed by \citet{Karlin67}.

Having the first structural bound for sequences of independent
identically distributed random variables, we derive an upper bound for
the differential Heaps law.
\begin{theorem}
  \label{theoUpperHeaps}
  For an IID process $(K_t)_{t\in\mathbb{Z}}$ over natural numbers,
  let $p_k:=P(K_0=k)$. Suppose that the structural inequality
  (\ref{UpperStructure}) holds for some constants $\beta\in(0,1)$,
  $A_S^+>0$, and all $p>0$. Then for
  \begin{align}
    A_V^+=\Gamma(1-\beta)A_S^+,
  \end{align}
  we have an upper bound for the differential Heaps law
  \begin{align}
    \label{UpperHeaps}
    \frac{V(\mathcal{K}_{t+1}^{t+s}\setminus\mathcal{K}_{1}^{t})}{s}
    &\le \frac{A_V^+}{t^{1-\beta}}.
  \end{align}
\end{theorem}

By contrast, if the process is stationary with the second structural
bound and appropriately bounded conditional probabilities, then a
lower bound for the differential Heaps law is satisfied.
\begin{theorem}
  \label{theoLowerHeaps}
  For a stationary process $(K_t)_{t\in\mathbb{Z}}$ over natural
  numbers, let $p_k:=P(K_0=k)$ and $p_k(t):=P(K_t=k|K_0=k)$. Suppose
  that the structural inequality (\ref{LowerStructure}) and 
  condition
  \begin{align}
    \label{ConditionalBound}
    p_k(t)&\le A_M p_k^\delta
  \end{align}
  hold for some constants $\beta\in(0,1)$, $\delta\in(0,1]$,
  $A_S^-,A_M>0$, all $k,t\in\mathbb{N}$, and all $p>0$. Then for
  \begin{align}
    A_V^-:=\frac{A_S^-}{2}\okra{\frac{1}{4A_M}}^{\frac{1-\beta}{\delta}},
  \end{align}
  we have a lower bound for the differential Heaps law
  \begin{align}
    \label{LowerHeaps}
    \frac{V(\mathcal{K}_{t+1}^{t+s}\setminus\mathcal{K}_{1}^{t})}{s}
    &\ge \frac{A_V^-}{(t+s)^{\frac{1-\beta}{\delta}}}.
  \end{align}
\end{theorem}
For IID processes, condition (\ref{ConditionalBound}) holds with
$A_M=1$ and $\delta=1$. For finite-state processes, it holds with
$A_M<\infty$ and also $\delta=1$, as it will be discussed further.

The lower bound for the differential Heaps law implies the analogous
bound for the differential Hilberg law in case of some stationary
Santa Fe processes.
\begin{theorem}
  \label{theoHilbergSantaFe}
  Consider a process $(X_t)_{t\in\mathbb{Z}}$ and a one-to-one mapping
  $\phi$ such that there holds a Santa Fe decomposition
  \begin{align}
    \label{SantaFeDisguised}
    \phi(X_t)=(K_t,Z_{K_t}),
  \end{align}
  where narration $(K_t)_{t\in\mathbb{Z}}$ is a stationary ergodic
  process over natural numbers and knowledge $(Z_k)_{k\in\mathbb{N}}$
  is a sequence of random variables independent of narration
  $(K_t)_{t\in\mathbb{Z}}$ with conditional entropies
  \begin{align}
    \label{KnowledgeEntropySantaFe}
    H(Z_k|Z_1^{k-1})\ge A_K.
  \end{align}
  for a constant $A_K>0$. If the marginal entropy satisfies
  \begin{align}
  \label{EntropyCond}
    H(X_t)<\infty
  \end{align}
  then we have inequality
  \begin{align}
    \label{HilbergConditionalSantaFe}
    \frac{H(X_{k+1}^{k+s}|X_{1}^{t})}{s}-h
    &\ge
      \frac{A_KV(\mathcal{K}_{k+1}^{k+s}\setminus\mathcal{K}_{1}^{t})}{s},
  \end{align}
  where $h$ is the entropy rate defined as (\ref{EntropyRateInf}).
\end{theorem}
Inequality (\ref{EntropyCond}) holds for all stationary processes over
a finite alphabet.

Finally, using the lower bound for the differential Hilberg law, we
derive a special lower bound for the neural scaling law.
\begin{theorem}
  \label{theoHilbergNeural}
  Let $(X_t)_{t\in\mathbb{Z}}$ be an arbitrary process such that we
  have a lower bound for the differential Hilberg law
  \begin{align}
  \label{HilbergConditional}
    \sup_{k\ge t}
    \frac{H(X_{k+1}^{k+s}|X_{1}^{t})}{s}-h
  &\ge \frac{A_H^-}{(t+s)^{1-\beta}}
  \end{align}
  for some constants $h\in(0,\infty)$, $\beta\in(0,1)$, and $A_H^->0$.
  Let $Q$ be a random probability measure that satisfies bounds
  \begin{align}
    \label{ParameterBound}
    H(Q)&\le A_H^-n,
    \\
    \label{ComputeBound}
    H(Q|X_1^t)&\le A_H^-c.
  \end{align}
  Define also the maximal test sample length
  \begin{align}
    \label{TestBound}
    s_{\max}(t,n,c)
    &:=\min\klam{
      \frac{t\sqrt{\frac{ct^{-\beta}}{1-\beta}}}{1-\sqrt{\frac{ct^{-\beta}}{1-\beta}}}
      ,
      \okra{\frac{n}{1-\beta}}^{\frac{1}{\beta}}
      }.
  \end{align}  
  Then for test sample lengths $s\le s_{\max}(t,n,c)$ we have the
  lower bound for the neural scaling
  \begin{align}
    &\sup_{k\ge t}
      \frac{\mean(-\log Q(X_{k+1}^{k+s}))}{s}-h
      \nonumber\\
    &\ge
      A_H^-
      \max\klam{
      \frac{1}{t^{1-\beta}}
      \okra{\frac{1-\sqrt{\frac{ct^{-\beta}}{1-\beta}}}{1+\sqrt{\frac{ct^{-\beta}}{1-\beta}}}}^{1-\beta}
      -\frac{c}{t}\okra{\frac{1-\sqrt{\frac{ct^{-\beta}}{1-\beta}}}{2\sqrt{\frac{ct^{-\beta}}{1-\beta}}}}
      ,
      \frac{2^{\beta}-1+\beta}{2}
      \okra{\frac{1-\beta}{n}}^{\frac{1-\beta}{\beta}}
      }
      .
      \label{NeuralConditional}
  \end{align}
\end{theorem}

\subsection{Relevance of the assumptions}
\label{secRelevance}

Overall, our results support a view in which neural scaling laws are
not isolated empirical phenomena but the final link in a broader chain
of statistical regularities that connect token frequencies, vocabulary
growth, long-range predictability, and learnability.  Let us discuss
the relevance of assumptions of Theorems
\ref{theoSandwichZipf}--\ref{theoHilbergNeural}.

\paragraph{Long-range dependence.}
Theorems \ref{theoSandwichZipf}--\ref{theoLowerHeaps} deal with a
subclass of stationary processes.  We notice that for $\delta=1$,
condition (\ref{ConditionalBound}) reduces to a uniform upper bound on
the pointwise mutual information between any two tokens. That
condition is satisfied for IID and finite-state processes
\citep[Theorem 11.3]{Debowski21}.  It is implied by condition
$\psi^*(1)<\infty$ in the terminology of \citet{Bradley05}. We also
note that any stationary strongly mixing process
$(K_t)_{t\in\mathbb{Z}}$ satisfies identity
\begin{align}
  \lim_{t\to\infty} p_k(t)=p_k.
\end{align}
By contrast, the proof of Theorem \ref{theoLowerHeaps} allows for
taking also $\delta<1$, which may be useful for long-range dependent
processes. In particular, natural language may require $\delta<1$
because of Taylor's law \citep{KobayashiTanakaIshii18, TanakaIshii21}
and burstiness of words \citep{AltmannPierrehumbertMotter09}.

\paragraph{Validity of the Santa Fe decomposition.}
Theorem \ref{theoHilbergSantaFe} deals with another subclass of
stationary processes and is a proxy for more general results. Its
assumptions are motivated by the intuition that language and other
information-bearing signals communicate discrete units of
knowledge. In the Santa Fe decomposition, the random variables
$(Z_k)_{k\in\mathbb{N}}$ represent such units, whereas the process
$(K_t)_{t\in\mathbb{Z}}$ determines which unit is being referred to at
a given position. The assumption (\ref{KnowledgeEntropySantaFe})
states that each unit contributes a non-negligible amount of novel
information and therefore cannot be reconstructed from previous units.

Although this representation is highly idealized, it captures a
natural distinction between references and facts. For example, words,
phrases, or textual passages often refer repeatedly to the same
underlying entities, events, or concepts. In this interpretation,
narration $(K_t)_{t\in\mathbb{Z}}$ determines the sequence of
references, whereas knowledge $(Z_k)_{k\in\mathbb{N}}$ describes the
latent facts being referenced. The growth of predictive information
then reflects the gradual discovery of previously unknown facts rather
than merely the repetition of surface forms.

The Santa Fe decomposition should therefore be viewed not as a literal
generative model of language but as a mathematical abstraction of
semantic structure. Establishing analogous results under less
restrictive assumptions is an important open problem.  One would like
to replace the explicit decomposition (\ref{SantaFeDisguised}) by more
intrinsic information-theoretic conditions that characterize the
existence of persistent latent variables or reusable knowledge
fragments. Such a development can extend the connection between Heaps'
and Hilberg's laws far beyond the specific class of Santa Fe
processes.

In view of so called theorems about facts and words \citep[Section
8.4]{Debowski21}, we suppose that one can replace decomposition
(\ref{SantaFeDisguised}) in Theorem \ref{theoHilbergSantaFe} by a more
general assumption such as strong non-ergodicity. According to
\citep[Definition 5.16]{Debowski21}, a stationary process is called
strongly non-ergodic if its shift-invariant $\sigma$-field is
non-atomic. Equivalently, it means that there is a continuous real
random variable that can be estimated arbitrarily well from any
sufficiently long sample \citep[Theorem 8.12]{Debowski21}. This is a
natural generalization since Santa Fe processes are strongly
non-ergodic. We also note that non-ergodicity is not a necessary
condition to obtain Hilberg's law. There are simple modifications of
Santa Fe processes and other constructive toy examples that allow for
Hilberg's law in mixing or non-mixing ergodic settings, see
\citep{Debowski12, Debowski20, Debowski21b, Debowski23}.  We need more
examples and results in this vein, also for numerical experiments with
artificial data, but their detailed discussion exceeds a reasonable
page limit of this article.

\paragraph{Hallucinated worlds vs.\ arbitrary words.}
As we can see, after discerning a general implication from Hilberg's
law to neural scaling, the burden of explanations shifts to stating
why Hilberg's hypothesis may be sound. In some settings, Hilberg's
hypothesis can be viewed as Zipf's law for tokens that are internally
random enough like in the Santa Fe processes (\ref{SantaFe}) or their
lossless encoding (\ref{SantaFeDisguised}).  The concrete
interpretation of these tokens can be varied. In Section
\ref{secPowerLaws}, we suggested a semantic interpretation, where
random variables (\ref{SantaFe}) are logically consistent propositions
that describe a certain random state of world. This world can be
``coherently hallucinated'', so as to speak in modern terms. Such
interpretation dates back to \citet{Debowski11b} and was defended
independently by \citet{MichaudOther23}.

Yet, there is a more simple-minded interpretation.  Decomposition
(\ref{SantaFeDisguised}), somewhat more abstract than (\ref{SantaFe}),
suggests that Hilberg's law is equivalent to Zipf's law for
orthographic words if those have sufficiently arbitrary
shapes. Arbitrariness of word shapes is one of the classical tenets of
linguistics \citep{DingemanseOther15}. However, proving that Hilberg's
law is tantamount to Zipf's law for empirically detectable and partly
random word-like structures requires an inequality opposite to
(\ref{HilbergConditionalSantaFe}). Turning
(\ref{HilbergConditionalSantaFe}) into a sandwich bound necessitates a
longer excursion to universal coding and involves results that exceed
our present statement of so called theorems about facts and words
\citep[Section 8.4]{Debowski21}, see also \citep{Debowski11b,
  Debowski21b}.  For this reason, we postpone this theoretical
development to another article. We want to stress, however, that the
explanation of Hilberg's hypothesis by arbitrariness of word shapes
does not exclude a semantic interpretation of this law if the
information content of the vocabulary and the information content of
other more abstract knowledge are of a similar magnitude.  A
one-to-one correspondence between words and concepts may be the
desired optimum.

\paragraph{Bound on the parameter count.}
The entropy of a discrete object cannot be essentially larger than the
description length of this object. Thus limiting the amount of a
certain resource implies constraining the respective entropy. In view
of this, bound (\ref{ParameterBound}) on the entropy of parameters
seems uncontroversial. However, the reality may be more complicated.
Observation $\gamma_N>\gamma_T$ for the simple parameter count
reported by \citet{KaplanOther20} seems to contradict Theorem
\ref{theoHilbergNeural} in contrast to inequality $\gamma_N<\gamma_T$
reported by \citet{HoffmannOther22}.  We note that these inconsistent
observations might be potentially explained out because the number of
trainable parameters is an imperfect proxy for the amount of
information represented by a model.  A real-valued parameter can in
principle encode an arbitrarily large number of bits and different
foundation models may use a varying amount of digits of the binary
expansion per real-valued parameter \citep{LiuLiuGore25}. Thus a naive
parameter count may overestimate or underestimate the effective
capacity of a real-valued neural network and one should rather use
universally comparable measures of model complexity.  Determining the
true number of degrees of freedom of a foundation model may not be so
straightforward.

\paragraph{Bound on compute.}
Bound (\ref{ComputeBound}) is also formulated in a purely
information-theoretic sense but its justification is more fragile.  We
model the time complexity of computing model $Q$ from training data
$X_1^t$ by conditional entropy, whereas using time-bounded Kolmogorov
complexity \citep{Levin73, Levin84} should be more proper for this
goal.  Inequality (\ref{ComputeBound}) may make sense if $Q$ is a
stochastic function of $X_1^t$, which is the less predictable from
$X_1^t$ the more compute is applied. Some ideas of
\citet{AchilleSoatto26} may be useful to rectify this issue, at the
expense of significantly complicating the proof of the time-bounded
analogue of Theorem \ref{theoHilbergNeural}. To keep this paper
focused on purely Shannon information measures, we relegate this
problem to future research.

\paragraph{The worst-case expected loss.}
Theorem \ref{theoHilbergNeural} suggests that the natural theoretical
definition of the model cross entropy for the neural scaling laws
(\ref{NeuralT})--(\ref{NeuralC}) in long-range dependent processes is
\begin{align}
  \label{CrossSup}
  h(s,t,n,c):=\sup_{k\ge t}\frac{\mean(-\log Q(X_{k+1}^{k+s}))}{s}.
\end{align}
where the foundation model $Q$ has been tested on data $X_{k+1}^{k+s}$
and trained on data $X_1^t$ with the amount of parameters $n$ and the
amount of compute $c$.  In this setting, $h(s,t,n,c)$ becomes a sort
of the worst-case expected cross entropy rate. Actually, setting
(\ref{CrossSup}) may model desirable testing conditions theoretically
better than using
\begin{align}
  \label{CrossNext}
  h(s,t,n,c):=\frac{\mean(-\log Q(X_{t+1}^{t+s}))}{s},
\end{align}
since one usually wants the test data to be maximally different from
the training data while still coming from the same generative source.
We notice that the difference between (\ref{CrossSup}) and
(\ref{CrossNext}) does not arise for memoryless sources but it can be
quite important for general long-range dependent stationary or
non-stationary data.

\paragraph{Inequalities and interactions.}
Because a more appropriate information-theoretic bound on compute is
needed, determining a more precise relationship between the Hilberg
exponent and the neural scaling exponents remains an important
empirical and theoretical problem. Still, let us make the following
remark.  Suppose that conditions (\ref{NeuralT})--(\ref{NeuralC}) hold
with a finite compute bound $c\ll t^\beta$ rather than $c=\infty$. If
additionally the neural scaling law bound (\ref{NeuralConditional}) is
true then we obtain asymptotic inequalities
\begin{align}
\gamma_T &\le 1-\beta,
&
\gamma_N &\le \frac{1-\beta}{\beta},
\end{align}
which generalize equalities (\ref{PredictedExponents}) known for
exchangeable Santa Fe processes (those with an IID narration). Besides
these inequalities, which hold in the general setting of arbitrary
processes, Theorem \ref{theoHilbergNeural} predicts quite complicated
interactions between parameters $t$, $n$, and $c$ and the cross
entropy of the foundation model. Having such a theoretical baseline
may be more informative for the research of neural scaling laws than
fitting theory-free regressions to empirical data.

\medskip%
This remark ends our discussion. After going through the paper
conclusion, the reader is invited to consult Appendices
\ref{secBitsTypes} and \ref{secImplications}. The appendices contain
the main mathematical derivations as well as some general results of
quantitative linguistic importance.

\section{Conclusion}
\label{secConclusion}

Extending prior results in machine learning such as \citep{Hutter21,
  MichaudOther23} and earlier results in information theory and
quantitative linguistics \citep{Debowski06, Debowski11b, Debowski21,
  Debowski23}, we have supplied and inspected a deductive chain that
connects Zipf's law to neural scaling. By isolating the discrete steps
that go through Heaps' law and Hilberg's hypothesis and by giving
explicit assumptions needed for each step, we have attempted to
clarify which aspects of natural language are responsible for the
power-law behavior observed in foundation models.  Our results show
that once vocabulary growth exhibits a power-law growth and once block
entropy inherits this scaling, then the constraints imposed by limited
data, parameters, and compute produce the neural scaling law.

The direct reduction of the neural scaling law to Hilberg's hypothesis
shifts attention to the deeper question of why power-law scaling of
entropy arises in natural language at all.  Broadly speaking, the
results of this paper suggest that statistical laws of language and
scaling laws of learning should not be studied in isolation. A
comprehensive theory would explain not only individual laws but also
the relationships among them, the numerical values of their exponents,
and their variation.  We hope that the framework developed here
provides a baseline for such investigation and that future work will
refine the idealizations which we have identified.

\section*{Acknowledgments and Disclosure of Funding}

Several paragraphs in the Introduction and in the Conclusion were
drafted in a dialog between the author and ChatGPT
(\url{https://chatgpt.com/}) and critically post-edited by the
author. Subsequently, the article was reviewed by the Stanford Agentic
Reviewer (\url{https://paperreview.ai/}) and its suggestions were
partially followed by the author.

This work received no external funding.
  
\appendix
  
\section{Remarks on counting bits and types}
\label{secBitsTypes}

This appendix is a gentle mathematical introduction that precedes the
proper proofs of Theorems
\ref{theoSandwichZipf}--\ref{theoHilbergNeural}.  It is split into
four parts. %
Part~\ref{secInequalities} treats general analogies for Shannon
entropy and expected cardinality. %
Part~\ref{secArbitrary} discusses the rates of entropy and expected
cardinality for arbitrary (non-stationary) processes. %
Part~\ref{secStationary} handles the block entropy and the expected
block cardinality for stationary processes. It also introduces
spectrum elements and bounds the rate of hapaxes. %
Part~\ref{secIID} analyzes the spectrum elements for IID processes. %

\subsection{Fundamental inequalities}
\label{secInequalities}

We will be developing parallel results for the Shannon entropy and the
expected number of types. Our reasoning is based on algebraically
simple but systematically refined information-theoretic
considerations. The general spirit of the $I$-measure will be
accompanying them \citep{Hu62}. For the textbook treatment and the
background, we refer to \citep{CoverThomas06,Yeung02}.
% The generalization of Shannon information measures to arbitrary
% $\sigma$-fields is also a tool that we shall use \citep[Sections 5.3
% and 5.4]{Debowski21}.

In general, both Shannon entropy and expected cardinality are
subadditive and enjoy the triangle inequality. The formulas for
Shannon entropy are well known.
\begin{proposition}[subadditivity]
  For random variables $X,Y,Q$, we have
\begin{align}
  \label{Subadditive}
  H(X,Y|Q)&\le H(X|Q)+H(Y|Q).
\end{align}
\end{proposition}
\begin{proof}
  The claim follows by the chain rule
  \begin{align}
    H(X,Y|Q)&= H(X|Q)+H(Y|Q,X)
  \end{align}
  and inequality $H(Y|Q,X)\le H(Y|Q)$.
\end{proof}
\begin{proposition}[triangle inequality]
  For random variables $X,Y,Q$, we have
\begin{align}
  \label{Triangle}
  H(X|Y)&\le H(X|Q)+H(Q|Y).
\end{align}  
\end{proposition}
\begin{proof}
  The claim follows by the chain rule
  \begin{align}
    H(X,Q|Y)&= H(X|Y,Q)+H(Q|Y)
  \end{align}
  and inequalities $H(X|Y)\le H(X,Q|Y)$ and $H(X|Y,Q)\le H(X|Q)$.
\end{proof}

The above familiar formulas for Shannon entropy have their
counterparts for expected cardinality if we replace pairing of
variables with the union of sets and conditioning with the set
difference.
\begin{proposition}[subadditivity]
  For random sets $\mathcal{X},\mathcal{Y},\mathcal{Q}$, we have
\begin{align}
  \label{SubadditiveV}
  V(\mathcal{X}\cup \mathcal{Y}\setminus \mathcal{Q})
  &\le V(\mathcal{X}\setminus \mathcal{Q})+V(\mathcal{Y}\setminus \mathcal{Q}).
\end{align}
\end{proposition}
\begin{proof}
  The claim follows by the chain rule
  \begin{align}
    V(\mathcal{X}\cup \mathcal{Y}\setminus \mathcal{Q})
    &= V(\mathcal{X}\setminus \mathcal{Q})+V(\mathcal{Y}\setminus \mathcal{Q}\cup \mathcal{X})
  \end{align}
  and inequality $V(\mathcal{Y}\setminus \mathcal{Q}\cup \mathcal{X})\le V(\mathcal{Y}\setminus \mathcal{Q})$.
\end{proof}
\begin{proposition}[triangle inequality]
  For random sets $\mathcal{X},\mathcal{Y},\mathcal{Q}$, we have
\begin{align}
  \label{TriangleV}
  V(\mathcal{X}\setminus \mathcal{Y})
  &\le V(\mathcal{X}\setminus \mathcal{Q})+V(\mathcal{Q}\setminus \mathcal{Y}).
\end{align}  
\end{proposition}
\begin{proof}
  The claim follows by the chain rule
  \begin{align}
    V(\mathcal{X}\cup \mathcal{Q}\setminus \mathcal{Y})
    &= V(\mathcal{X}\setminus \mathcal{Y}\cup \mathcal{Q})+V(\mathcal{Q}\setminus \mathcal{Y})
  \end{align}
  and inequalities
  $V(\mathcal{X}\setminus \mathcal{Y})\le V(\mathcal{X}\cup
  \mathcal{Q}\setminus \mathcal{Y})$ and
  $V(\mathcal{X}\setminus \mathcal{Y}\cup \mathcal{Q})\le
  V(\mathcal{X}\setminus \mathcal{Q})$.
\end{proof}

There is also an important bridging inequality for cross entropy of
random measures.
\begin{proposition}[source coding]
  For a random probability measure $Q$ applied to another
  random variable $X$, we have
  \begin{align}
  \label{SourceCoding}
  \mean(-\log Q(X))&\ge H(X|Q).
  \end{align}
\end{proposition}
\begin{proof}
  We have
  \begin{align}
    \mean\okra{-\log Q(X)}
    &=\mean\mean\okra{-\log Q(X)\middle|Q}
      \nonumber\\
    &=\mean\mean\okra{-\log P(X|Q)\middle|Q}
      +\mean\mean\okra{\log \frac{P(X|Q)}{Q(X)}\middle|Q}
      \nonumber\\
    &\ge \mean\mean\okra{-\log P(X|Q)\middle|Q}
      \nonumber\\
    &=\mean\okra{-\log P(X|Q)}
      =H(X|Q),
      %\qedhere
  \end{align}
  where we used the law of total expectation and the fact that the
  Kullback-Leibler divergence is non-negative.
\end{proof}

\subsection{Arbitrary processes}
\label{secArbitrary}

In the following, we consider arbitrary processes (possibly
non-stationary) over a countable alphabet.  Let us inspect the rates
of the Shannon entropy and the expected cardinality for these
objects. We may define the rates as follows.
\begin{definition}
  Let $(X_t)_{t\in\mathbb{N}}$ be an arbitrary process.  We
  define the entropy rate
  \begin{align}
    \label{SupEntropy}
    h
    &:=
      \inf_{s\in\mathbb{N}}\sup_{k\ge 0} \frac{H(X_{k+1}^{k+s})}{s}.
  \end{align}  
\end{definition}
\begin{definition}
  Let $(K_t)_{t\in\mathbb{N}}$ be an arbitrary process.  We define the
  expected cardinality rate
  \begin{align}
    \label{SupTypes}
    v
    &:=
      \inf_{s\in\mathbb{N}} \sup_{k\ge 0} \frac{V(\mathcal{K}_{k+1}^{k+s})}{s}.
  \end{align}  
\end{definition}

The Fekete lemma \citep{Fekete23} states that any sequence
$(a_s)_{s\in\mathbb{N}}$ that satisfies the condition of subadditivity
$a_{s+r}\le a_s+a_r$ has the rate that can be equivalently defined as
\begin{align}
  \label{Fekete}
  \lim_{s\to\infty}\frac{a_s}{s}=\inf_{s\in\mathbb{N}}\frac{a_s}{s}.
\end{align}
  Hence having subadditivity (\ref{Subadditive}) and
  (\ref{SubadditiveV}), we can show that rates (\ref{SupEntropy}) and
  (\ref{SupTypes}) can be equivalently expressed somewhat differently.
\begin{proposition}
  \label{theoSupEntropy}
  Let $(X_t)_{t\in\mathbb{N}}$ be a process such that $h<\infty$ for
  the entropy rate (\ref{SupEntropy}).  For an arbitrary $t\ge 0$, we
  have
  \begin{align}
    h
    =\lim_{s\to\infty}\sup_{k\ge t} \frac{H(X_{k+1}^{k+s}|X_{1}^{t})}{s}
    =\inf_{s\in\mathbb{N}}\sup_{k\ge t} \frac{H(X_{k+1}^{k+s}|X_{1}^{t})}{s}.
  \end{align}
\end{proposition}
\noindent
\begin{proof}
  By inequality (\ref{Subadditive}), we notice subadditivity
  \begin{align}
    \sup_{k\ge t} H(X_{k+1}^{k+s+r}|X_{1}^{t})
    &\le
      \sup_{k\ge t}
      \kwad{H(X_{k+1}^{k+s}|X_{1}^{t})
      +
      H(X_{k+s+1}^{k+s+r}|X_{1}^{t})}
      \nonumber\\
    &\le
      \sup_{k\ge t}
      H(X_{k+1}^{k+s}|X_{1}^{t})
      +
      \sup_{k\ge t}
      H(X_{k+1}^{k+r}|X_{1}^{t}).
  \end{align}
  Hence by the Fekete lemma \citep{Fekete23}, we obtain
  \begin{align}
    h(t)
    :=\lim_{s\to\infty}\sup_{k\ge t} \frac{H(X_{k+1}^{k+s}|X_{1}^{t})}{s}
    =\inf_{s\in\mathbb{N}}\sup_{k\ge t} \frac{H(X_{k+1}^{k+s}|X_{1}^{t})}{s}.
  \end{align}
  It suffices to show $h(t)=h$.  Since $h<\infty$, we have
  $\sup_{t\ge 0} H(X_t)<\infty$. Hence, by the chain rule and a simple
  calculation, we obtain a uniform bound
  \begin{align}
    \abs{H(X_{k+1}^{k+s})-H(X_{k+t+1}^{k+t+s})}\le B(t)<\infty.
  \end{align}
  For the same reason, we also have
  \begin{align}
   \abs{H(X_{k+t+1}^{k+t+s})-H(X_{k+t+1}^{k+t+s}|X_{1}^{t})}\le D(t)<\infty.
  \end{align}
  Chaining these two sandwich bounds and taking the supremums over $k$
  and infimums over $s$, we infer $h(t)=h$.
\end{proof}
\begin{proposition}
  \label{theoSupTypes}
  Let $(K_t)_{t\in\mathbb{N}}$ be an arbitrary process. For an
  arbitrary $t\ge 0$, we have
  \begin{align}
    v
    =\lim_{s\to\infty} \sup_{k\ge t}
    \frac{V(\mathcal{K}_{k+1}^{k+s}\setminus\mathcal{K}_{1}^{t})}{s}
    =\inf_{s\in\mathbb{N}} \sup_{k\ge t}
    \frac{V(\mathcal{K}_{k+1}^{k+s}\setminus\mathcal{K}_{1}^{t})}{s}.
  \end{align}
\end{proposition}
\noindent
\begin{proof}
  Mutatis mutandis, the same as the proof of Proposition
  \ref{theoSupEntropy}.
\end{proof}

In view of the above two propositions, we have
\begin{align}
  \sup_{k\ge t} H(X_{k+1}^{k+s}|X_{1}^{t})&\ge hs,
  \\
  \sup_{k\ge t} V(\mathcal{K}_{k+1}^{k+s}\setminus\mathcal{K}_{1}^{t})&\ge vs.
\end{align}

\subsection{Stationary processes}
\label{secStationary}

In case of stationary processes, the translation symmetry yields some
further results. By a general result, any stationary process with
natural number indices can be uniquely extended to the stationary
process with integer indices \citep[Lemma 9.2]{Kallenberg97}. Thus,
for stationary processes $(X_t)_{t\in\mathbb{Z}}$ and
$(K_t)_{t\in\mathbb{Z}}$, we denote the the block Shannon entropy and
the expected number of types
\begin{align}
  H(t)&:=H(X_1^t),
  \\
  V(t)&:=V(\mathcal{K}_{1}^{t}).
\end{align}
Let us define the finite difference operator
$\Delta f(t):=f(t+1)-f(t)$. Function $f(t)$ is called positive,
growing, and concave if $f(t)\ge 0$, $\Delta f(t)\ge 0$, and
$\Delta^2 f(t)\le 0$, respectively.

We have two analogous statements. We recall the results for entropy to
apply them by analogy to the expected cardinality.
\begin{proposition}
  \label{theoBlockEntropy}
  Let a stationary process $(X_t)_{t\in\mathbb{Z}}$.  For
  $t\in\mathbb{N}$, we claim that:
  \begin{enumerate}   
  \item Function $t\mapsto H(t)$ is positive, growing, and concave and
    $H(0)=0$.
  \item Function $t\mapsto H(t)/t$ is decreasing.
  \item We have $h=\lim_{t\to\infty}H(t)/t\ge 0$.
  \end{enumerate}  
\end{proposition}
\begin{proof}
  See \citep[Section 5.2]{Debowski21}. In general, we have
  \begin{align}
      H(t)
      &=H(X_1^t)\ge 0,
      \\
      \Delta H(t)
      &=H(X_{t+1}|X_1^t)\ge 0,
      \\
      \Delta^2 H(t)
      &=-I(X_0;X_{t+1}|X_1^t)\le 0.
  \end{align}
  Function $t\mapsto \Delta H(t)$ is decreasing since $\Delta^2 H(t)\le
  0$. Hence
    \begin{align}
      \frac{H(t+1)}{t+1}
      =
      \frac{\sum_{i=0}^{t}\Delta H(i)}{t+1}
      &\le
      \frac{
      \sum_{i=0}^{t-1}\Delta H(i)+\frac{\sum_{i=0}^{t-1}\Delta H(i)}{t}
        }{t+1}
        \nonumber\\
      &=
      \frac{\sum_{i=0}^{t-1}\Delta H(i)}{t}
      =
      \frac{H(t)}{t}.
    \end{align}
    Thus function $t\mapsto H(t)/t$ is decreasing and limit
    $\lim_{t\to\infty}H(t)/t$ exists.  By stationarity, it equals $h$
    defined in (\ref{SupEntropy}).
\end{proof}
\begin{proposition}
  \label{theoBlockTypes}
  Let a stationary process $(K_t)_{t\in\mathbb{Z}}$.  For
  $t\in\mathbb{N}$, we claim that:
  \begin{enumerate}   
  \item Function $t\mapsto V(t)$ is positive, growing, and concave and
    $V(0)=0$.
  \item Function $t\mapsto V(t)/t$ is decreasing and $V(t)/t\le 1$.
  \item We have $v=\lim_{t\to\infty} V(t)/t=0$.
  \end{enumerate}  
\end{proposition}
\begin{proof}
  The proof is analogous to the proof of Proposition
  \ref{theoBlockEntropy} except for the last claim. In particular, we
  have
   \begin{align}
      V(t)
      &=V(\mathcal{K}_1^t)\ge 0,
      \\
      \Delta V(t)
      &=V(\mathcal{K}_{t+1}\setminus\mathcal{K}_1^t)\ge 0,
      \\
      \Delta^2 V(t)
      &=-V(\mathcal{K}_0\cap\mathcal{K}_{t+1}\setminus\mathcal{K}_1^t)\le 0.
   \end{align}
   Inequality $V(t)/t\le 1$ is obvious since the number of type is not
   greater than the number tokens.  The proof of
   $\lim_{t\to\infty} V(t)/t=0$ is as follows.  Without loss of
   generality, we assume that the alphabet is the set of natural
   numbers.  Consider a real function $g(t)\ge 0$.  Generalizing an
   idea used by \citet{Khmaladze88} for IID processes, we observe
  \begin{align}
    V(t)
    =
    \mean \sum_{k=1}^\infty \boole{k\in\mathcal{K}_1^t}
    =
    \sum_{k=1}^\infty P(k\in\mathcal{K}_1^t)
    \le
    g(t)+\sum_{k>g(t)} P(k\in\mathcal{K}_1^t),
  \end{align}
  whereas the union bound and stationarity yield
  \begin{align}
    P(k\in\mathcal{K}_1^t)
    =
    P(K_1=k\lor \ldots\lor K_t=k)
    \le
    tP(K_0=k).
  \end{align}
  Thus we have bound
  \begin{align}
    V(t) \le g(t)+tP(K_0>g(t))
  \end{align}
  that holds for any function $g(t)\ge 0$. In particular, for an
  $\epsilon>0$, we may take $g(t)=\epsilon t/2$. For all sufficiently
  large $t$, we observe $P(K_0>g(t))\le \epsilon/2$. Hence, for these
  $t$, we derive $V(t)/t \le \epsilon$. By arbitrariness of
  $\epsilon$, we derive $\lim_{t\to\infty} V(t)/t=0$.
\end{proof}

Denote the set of types that occur exactly $m$ times in sequence
$K_1^t$ as
\begin{align}
  \mathcal{K}_1^t(m):=\klam{k: \sum_{i=1}^t\boole{K_i=k}=m}.
\end{align}
Besides the expected total number of types $V(t)$, let us introduce the
spectrum elements
\begin{align}
  \label{Spectrum}
  V(t|m)=V(\mathcal{K}_1^t(m)),
\end{align}
being the expected numbers of types that occur exactly $m$ times
\citep{Karlin67, Khmaladze88, Baayen01}. The bar in notation $V(t|m)$
is a separator of arguments, not related to conditioning.  In
particular, we have $V(t|m)\ge 0$ and
$\sum_{m\in\mathbb{N}} V(t|m)=V(t)$ so the relative spectrum
$m\mapsto V(t|m)/V(t)$ is a probability distribution.

The first spectrum element $V(t|1)$ is the expected number of types
that occur once. These one-time elements are called \emph{hapax
  legomena} in Greek or, succinctly, hapaxes in English.  In the
stationary case, there is a general upper bound for the expected
number of hapaxes in terms of the difference of the total number of
types.
\begin{proposition}
  \label{theoHapaxes}
  For a stationary process $(K_t)_{t\in\mathbb{Z}}$, we have
  \begin{align}
    \label{Hapaxes}
    \frac{V(t|1)}{t}\le
    \Delta V\okra{\ceil{\frac{t}{2}}}.
  \end{align}
\end{proposition}
\begin{proof}
  For a stationary process, we observe that
  \begin{align}
    V(t|1)
    &=
      \sum_{i=1}^t
      P\okra{K_i\not\in\mathcal{K}_1^{i-1}
      \cup\mathcal{K}_{i+1}^t}
      \nonumber\\
    &\le
      \sum_{i=1}^t
      \min\klam{P\okra{K_i\not\in\mathcal{K}_1^{i-1}},
      P\okra{K_i\not\in\mathcal{K}_{i+1}^t}}
      \nonumber\\
    &=
      \sum_{i=1}^t
      \min\klam{\Delta V(i-1),\Delta V(t-i)}
      \le
      t\Delta V\okra{\ceil{\frac{t}{2}}}
  \end{align}
  since $t\mapsto\Delta V(t)$ is decreasing.
\end{proof}

\subsection{IID processes}
\label{secIID}

So far, the statements for Shannon entropy and expected cardinality
were mirror-like. However, for more specific processes, the analogy
between these two functionals is rather as follows: The expected
cardinality applies to IID processes in a similar fashion as the
Shannon entropy applies to exchangeable Santa Fe processes (those with
an IID narration). This analogy will be applied in Appendix
\ref{secHilbergSantaFe}. Now we consider the expected cardinality for
IID processes, being simpler to analyze.  We notice that the expected
cardinality for IID sources enjoys further special properties. Namely,
it is a Hausdorff sequence---a discrete-time analogue of a Bernstein
function. This observation nicely complements the results by
\citet{Karlin67} for IID and Poisson cases.

The development is as follows.
\begin{definition}
  A sequence $v:\mathbb{N}\cup\klam{0}\to\mathbb{R}$ is called a
  Hausdorff sequence if $v(t)\ge 0$ and $(-1)^{m+1}\Delta^m v(t)\ge 0$
  for all $m\in\mathbb{N}$ and $n\in\mathbb{N}\cup\klam{0}$.  A
  sequence $u:\mathbb{N}\cup\klam{0}\to\mathbb{R}$ is called
  completely alternating if $u(t)=\Delta v(t)$ for a certain Hausdorff
  sequence $v(t)$. We call a sequence $v(t)$ standard if $v(0)=0$ and
  $\Delta v(0)=1$.
\end{definition}
\noindent
\emph{Remark:}
The name \emph{Hausdorff sequence} is non-standard itself. We have
coined it by an analogy to the standard term \emph{Bernstein
  function}, which is the continuous time-analog, applying derivatives
rather than differences \citep{SchillingSongVondracek10}.

Any Hausdorff sequence can be expressed as a convex combination of
sequences $t\mapsto p^{-1}\okra{1-(1-p)^t}$ for varying $p>0$. This
result is known as the Hausdorff moment theorem \citep{Hausdorff23}.
It is a discrete-time analogue of the Bernstein theorem on completely
monotone functions \citep{Bernstein28}, also known as the
L\'evy-Khintchine representation in probability \citep{Levy48}.
\begin{proposition}
  \label{theoMoment}
  A sequence $v:\mathbb{N}\cup\klam{0}\to\mathbb{R}$ is a Hausdorff
  sequence if and only if there exists a unique non-negative measure
  $\tilde v$ on $[0,1]$ such that
  \begin{align}
    \label{FtoVD}
    v(t)&=t\tilde v(\klam{0})
          +
          \int_{(0,1)}\frac{1-(1-p)^t}{p}d\tilde v(p)
          +
          \tilde v(\klam{1}),
    \\
    \label{FtoVPrimeD}
    \Delta v(t)
        &=\tilde v(\klam{0})
          +
          \int_{(0,1)}(1-p)^t d\tilde v(p).
  \end{align}
\end{proposition}
\noindent
\emph{Remark:} We call measure $\tilde v$ the Hausdorff measure of
sequence $v(t)$. A Hausdorff sequence $v(t)$ is standard if and only
if $\tilde v(\klam{1})=0$ and the measure $\tilde v$ is a probability
measure.
\begin{proof}
  See \citep{Hausdorff23, Akhiezer21}.
\end{proof}

Now let us define another non-standard object. The bar in notation
$v(t\|m)$ is a separator of arguments, unrelated to conditioning.
\begin{definition}
  For an arbitrary sequence $v:\mathbb{N}\cup\klam{0}\to\mathbb{R}$,
  we define its Taylor elements
\begin{align}
  \label{Taylor}
  v(t\|m)
  &:=
    (-1)^{m+1}\binom{t}{m}\Delta^m v(t-m),
  & 1&\le m\le t.
\end{align}
\end{definition}
We notice that $v(t\|m)\ge 0$ if and only if sequence $v(t)$ is a
Hausdorff sequence.
\begin{proposition}
  The Taylor elements satisfy consistency conditions
  \begin{align}
    \label{Consistent}
    \sum_{m=1}^\infty v(t\|m)&=v(t),
    &
      \sum_{m=1}^\infty mv(t\|m)&=t. 
  \end{align}
  if and only if sequence $v(t)$ is standard.
\end{proposition}
\begin{proof}
  The claim follows by the Newton formula
  $(1-\Delta)^r=\sum_{m=0}^r \binom{r}{m} (-\Delta)^m$, written as
  \begin{align}
    \label{Newton}
    v(t-r)
    &=
      \sum_{m=0}^r (-1)^m
      \binom{r}{m} \Delta^m v(t-m),
  \end{align}
  which resembles the Taylor expansion. It suffices to consider
  (\ref{Newton}) for $r=t$ and $r=t-1$. In the first case, we obtain
  \begin{align}
    v(t)-\sum_{m=1}^t v(t\|m)
    =
    \sum_{m=0}^t (-1)^m
    \binom{t}{m} \Delta^m v(t)
    =
    v(t-t)
    =
    v(0)
    .
  \end{align}
  In the second case, we use identity
  $m\binom{t}{m}=t\binom{t-1}{m-1}$, which yields
  \begin{align}
    \sum_{m=1}^t mv(t\|m)
    &=
      \sum_{m=1}^t (-1)^{m+1}
      m\binom{t}{m} \Delta^m v(t)
      =
      t\sum_{m=1}^t (-1)^{m+1}
      \binom{t-1}{m-1} \Delta^m v(t)
      \nonumber\\
    &=
      t\sum_{j=0}^{t-1} (-1)^j\binom{t-1}{j} \Delta^{j+1} v(t)
      =
      t\Delta v(t-t+1)
      =
      t\Delta v(1)
      .
  \end{align}
  Hence conditions (\ref{Consistent}) hold if and only if sequence
  $v(t)$ is standard.
\end{proof}

Without loss of generality, let us assume that the alphabet of an IID
process $(K_t)_{t\in\mathbb{Z}}$ is the set of natural numbers.  By
the Bernoulli scheme, the expected number of types and the spectrum
elements defined via (\ref{Spectrum}) equal
\begin{align}
  V(t)
  &=\sum_{k\in\mathbb{N}} (1-(1-p_k)^{t}),
  \\
  \label{IIDSpectrum}
  V(t|m)
  &=
    \sum_{k\in\mathbb{N}} \binom{t}{m} p_k^{m}(1-p_k)^{t-m},
  & 1&\le m\le t,
\end{align}
where $p_k:=P(K_0=k)$. By formula (\ref{IIDSpectrum}) and identity
$\Delta (1-p_k)^{t}=-p_k(1-p_k)^{t}$, the spectrum elements $V(t|m)$
and the Taylor elements of $V(t)$ are equal,
\begin{align}
  \label{SpectrumTaylor}
  V(t|m)=V(t\|m).
\end{align}
Consequently, sequence $V(t)$ is a Hausdorff sequence. Moreover,
sequence $V(t)$ is standard and spectrum elements $V(t|m)$ satisfy
consistency conditions analogous to (\ref{Consistent}).  The Hausdorff
measure of $V(t)$ is an atomic probability measure and assumes form
\begin{align}
  \label{Points}
  \tilde V(A)=\sum_{k\in\mathbb{N}} p_k\boole{p_k\in A}.
\end{align}
By formulas (\ref{Taylor}) and (\ref{SpectrumTaylor}), the expected
rate of hapaxes for IID processes equals
\begin{align}
  \label{IIDHapaxes}
  \frac{V(t|1)}{t}=\Delta V(t-1),
\end{align}
which can be contrasted with the more general inequality
(\ref{Hapaxes}) for stationary sources.

\section{Proofs of the main results}
\label{secImplications}

It is high time to demonstrate the chain of implications
(\ref{Implications}).  Respectively, in
Appendices~\ref{secSandwichZipf}--\ref{secHilbergNeural}, we
demonstrate Theorems \ref{theoSandwichZipf}--\ref{theoHilbergNeural}.

\subsection{Proof of Theorem \ref{theoSandwichZipf}}
\label{secSandwichZipf}

Bounding the sums by integrals, may evaluate
\begin{align}
  \sum_{k\in\mathbb{N}} \boole{p_k\ge p}
  &\le \int_0^\infty \boole{A_Z^+k^{-1/\beta}\ge p}dk
    \nonumber\\
  &=\int_0^\infty\boole{k\le\okra{\frac{p}{A_Z^+}}^{-\beta}}dk
    =\okra{\frac{p}{A_Z^+}}^{-\beta},
  \\
  \sum_{k\in\mathbb{N}} p_k\boole{p_k\le p}
  &\ge \int_1^\infty A_Z^-k^{-1/\beta} \boole{A_Z^+k^{-1/\beta}\le p}dk
    \nonumber\\
  &= \int_1^\infty A_Z^-k^{-1/\beta} \boole{k\ge\okra{\frac{p}{A_Z^+}}^{-\beta}}dk
    % = \frac{A_Z^-}{1/\beta-1}
    % \okra{\okra{\frac{p}{A_Z^+}}^{-\beta}}^{-1/\beta+1}
      = \frac{A_Z^-}{1/\beta-1}\okra{\frac{p}{A_Z^+}}^{1-\beta}.
\end{align}

\subsection{Proof of Theorem \ref{theoUpperHeaps}}
\label{secUpperHeaps}

Since $t\mapsto\Delta V(t)$ is decreasing, we derive
\begin{align}
  \frac{V(\mathcal{K}_{t+1}^{t+s}\setminus\mathcal{K}_{1}^{t})}{s}
  &\le \Delta V(t).
\end{align}
By (\ref{IIDHapaxes}) and (\ref{IIDSpectrum}), we may bound
\begin{align}
  \Delta V(t)
  &=\frac{V(t+1|1)}{t+1}
    =\sum_{k\in\mathbb{N}} p_k(1-p_k)^t
    \nonumber\\
  &\le\sum_{k\in\mathbb{N}}\int_0^{p_k}(1-p)^tdp
    =\int_0^1 \okra{\sum_{k\in\mathbb{N}} \boole{p_k\ge p}}(1-p)^tdp
    \nonumber\\
  &\le A_S^+\int_0^1 p^{-\beta}(1-p)^tdp.
\end{align}
Further evaluation yields
\begin{align}
  \int_0^1(1-p)^t p^{-\beta}dp
  &=
    \frac{\Gamma(t+1)\Gamma(1-\beta)}{\Gamma(t+2-\beta)}
    % \nonumber\\
    % &=
    % \frac{t!}{(1-\beta)(2-\beta)\ldots(t+1-\beta)}
      \le \frac{\Gamma(1-\beta)}{t^{1-\beta}}.
\end{align}

\subsection{Proof of Theorem \ref{theoLowerHeaps}}
\label{secLowerHeaps}

Since $t\mapsto\Delta V(t)$ is decreasing, we derive
\begin{align}
  \frac{V(\mathcal{K}_{t+1}^{t+s}\setminus\mathcal{K}_{1}^{t})}{s}
  &\ge \Delta V(t+s).
\end{align}
By (\ref{Hapaxes}), using the union bound, we may write
\begin{align}
  \Delta V(t)
  &\ge\frac{V(2t|1)}{2t}
    =\frac{1}{2t}\sum_{i=1}^{2t}    
    P\okra{K_i\not\in\mathcal{K}_1^{i-1}
    \cup\mathcal{K}_{i+1}^{2t}}
    \nonumber\\
  &=\sum_{k\in\mathbb{N}} 
    \frac{1}{2t}\sum_{i=1}^{2t}
    P\okra{K_i=k\not\in\mathcal{K}_1^{i-1}
    \cup\mathcal{K}_{i+1}^{2t}}
    \nonumber\\
  &=\sum_{k\in\mathbb{N}} 
    \max\klam{0,
    \frac{1}{2t}\sum_{i=1}^{2t}
    P\okra{K_i=k\not\in\mathcal{K}_1^{i-1}
    \cup\mathcal{K}_{i+1}^{2t}}}
    \nonumber\\
  &\ge\sum_{k\in\mathbb{N}} 
    \max\klam{0,
    2p_k-\frac{1}{2t}\sum_{i=1}^{2t}\sum_{j=1}^{2t} P(K_i=k,K_j=k)}
    \nonumber\\
  &\ge\sum_{k\in\mathbb{N}} p_k\max\klam{0,1-A_M (2t-1)p_k^\delta}
    \ge\sum_{k\in\mathbb{N}} p_k\max\klam{0,1-2A_M tp_k^\delta}
    \nonumber\\
  &\ge\frac{1}{2}\sum_{k\in\mathbb{N}} p_k\boole{p_k^\delta\le \frac{1}{4A_M t}}
    % \nonumber\\
   \ge \frac{A_S^-}{2}\okra{\frac{1}{4A_M}}^{\frac{1-\beta}{\delta}}.
\end{align}

\subsection{Proof of Theorem \ref{theoHilbergSantaFe}}
\label{secHilbergSantaFe}

Without loss of generality, we may assume $P(K_t=k)>0$ for all
$k\in\mathbb{N}$ since we may re-index sequence
$(Z_k)_{k\in\mathbb{N}}$ eliminating those variables $Z_k$ for which
$P(K_t=k)=0$. Then we proceed as follows.

The main idea of the proof is contained in this paragraph.  By
independence of $(K_t)_{t\in\mathbb{Z}}$ and $(Z_k)_{k\in\mathbb{N}}$
and condition (\ref{KnowledgeEntropySantaFe}), we may decompose
\begin{align}
  H(X_1^t)
  &=
    H(K_1^t)+H(X_1^t|K_1^t)
    \nonumber\\
  &=
    H(K_1^t)
    +H(\klam{(k,Z_k):k\in\mathcal{K}_{1}^{t}}|\mathcal{K}_{1}^{t})
    \nonumber\\
  &\ge
    H(K_1^t)+A_KV(\mathcal{K}_{1}^{t}).
\end{align}
Analogously, for $k\ge t$, we also obtain the conditional statement
\begin{align}
  H(X_{k+1}^{k+s}|X_{1}^{t})
  &=
    H(K_{k+1}^{k+s}|X_{1}^{t})+H(X_{k+1}^{k+s}|X_{1}^{t},K_{k+1}^{k+s})
    \nonumber\\
  &=
    H(K_{k+1}^{k+s}|K_{1}^{t})
    +H\okra{\klam{(k,Z_k):k\in\mathcal{K}_{k+1}^{k+s}\setminus\mathcal{K}_{1}^{t}}
    \middle|
    \mathcal{K}_{k+1}^{k+s}\setminus\mathcal{K}_{1}^{t}}
    \nonumber\\
  &\ge
    H(K_{k+1}^{k+s}|K_{1}^{t})
    +A_KV(\mathcal{K}_{k+1}^{k+s}\setminus\mathcal{K}_{1}^{t}).
\end{align}
Since conditioning decreases entropy, we derive
\begin{align}
  H(K_{k+1}^{k+s}|K_{1}^{t})\ge
  H(K_{k+1}^{k+s}|K_{-\infty}^{k})=sH(K_{t+1}|K_{-\infty}^{t})
\end{align}
by stationarity and the chain rule for conditional entropy generalized
to $\sigma$-fields \citep[Section 5.3]{Debowski21}. Hence we derive
(\ref{HilbergConditionalSantaFe}) if we prove the equality of
entropy rates
\begin{align}
  \label{EntropyRateEq}
  H(K_{t+1}|K_{-\infty}^{t})=H(X_{t+1}|X_{-\infty}^{t})=h.
\end{align}
Equality (\ref{EntropyRateEq}) seems intuitive but its formal proof is
not so short since there may arise some problems if
$H(X_t)=\infty$ and $H(X_t|X_{-\infty}^{t-1})=0$.

We prove (\ref{EntropyRateEq}) by applying the toolbox from
\citep{Debowski21} and the classical results from \citet{Breiman57}.
First, we demonstrate the left-most equality in (\ref{EntropyRateEq}).
For this goal, we observe the following facts:
\begin{itemize}
\item Since $P(K_t=k)>0$ for any $k\in\mathbb{N}$ and
  $(K_t)_{t\in\mathbb{Z}}$ is stationary ergodic then 
  $\lim_{t\to\infty} P(k\in\mathcal{K}_1^t)=1$ by the Birkhoff ergodic
  theorem \citep[Theorem 4.6]{Debowski21} and the Riesz theorem
  \citep[Theorem 3.27]{Debowski21}.
\item For each $k\in\mathbb{N}$, there is a function $g_k$ of strings
  such that for any $s\in\mathbb{Z}$ and $t\in\mathbb{N}$, we have
  $Z_k=g_k(X_{s+1}^{s+t})$ if $k\in\mathcal{K}_{s+1}^{s+t}$.

  \emph{In plain words, we know the correct value of $Z_k=z$ if we
    find a pair $(k,z)$ in $\phi(X_{s+1}^{s+t})$.}
\item As a result of two above points, by \citep[Theorem
  8.10]{Debowski21}, for each $k\in\mathbb{N}$, the completion of the
  $\sigma$-field of $Z_k$ is contained in the completion of the
  shift-invariant $\sigma$-field of $(X_t)_{t\in\mathbb{Z}}$.

  \emph{In plain words, it means that there is a function $g^*_k$ of
    infinite sequences such that for any $s\in\mathbb{Z}$, we have
    $Z_k=g^*_k((X_{t+s})_{t\in\mathbb{Z}})$ almost surely.}
\item On the other hand, we invoke that the completion of the
  shift-invariant $\sigma$-field is contained in the intersection of
  the completions of the tail $\sigma$-field of past and the tail
  $\sigma$-field of future of process $(X_t)_{t\in\mathbb{Z}}$
  \citep[Theorem 5.34]{Debowski21}.

  \emph{In plain words, it means that for any $t\in\mathbb{Z}$ there
    are functions $g^+_k$ and $g^-_k$ of semi-infinite sequences such
    that $Z_k=g^-_k(X_{-\infty}^{t})$ and $Z_k=g^+_k(X_{t}^{\infty})$
    almost surely.}

\end{itemize}
Hence, the completion of the $\sigma$-field of $X_{-\infty}^{t}$
contains the completion of the $\sigma$-field of $Z_1^\infty$.
 
The rest is a straightforward application of the calculus of Shannon
information measures for $\sigma$-fields \citep[Sections 5.3 and
5.4]{Debowski21}. We make the following observations:
\begin{itemize}
\item As it has been shown, the completion of the $\sigma$-field of
  $X_{-\infty}^{t}$ contains the completion of the $\sigma$-field of
  $Z_1^\infty$, so
\begin{align}
  H(X_{t+1}|X_{-\infty}^{t})=H(X_{t+1}|X_{-\infty}^{t},Z_1^\infty).
\end{align}
\item Since the $\sigma$-fields of $K_{-\infty}^{t+1}$ and $Z_1^\infty$
  are independent, we also have
\begin{align}
  H(K_{t+1}|K_{-\infty}^{t})=H(K_{t+1}|K_{-\infty}^{t},Z_1^\infty).
\end{align}
\item The $\sigma$-field of $(X_{-\infty}^{t},Z_1^\infty)$ equals the
  the $\sigma$-field of $(K_{-\infty}^{t},Z_1^\infty)$ by
  decomposition (\ref{SantaFeDisguised}).
\item Moreover, the $\sigma$-field of $Z_{K_{t+1}}$ is contained in
  the $\sigma$-field of $(K_{-\infty}^{t+1},Z_1^\infty)$, so
\begin{align}
  H(Z_{K_{t+1}}|K_{-\infty}^{t+1},Z_1^\infty)=0.
\end{align}
\end{itemize}
Hence by the generalized chain rule, we may write
\begin{align}
  H(X_{t+1}|X_{-\infty}^{t})
  &=H(X_{t+1}|X_{-\infty}^{t},Z_1^\infty)
    \nonumber\\
  &= H(X_{t+1}|K_{-\infty}^{t},Z_1^\infty)
    \nonumber\\
  &= H(K_{t+1}, Z_{K_{t+1}}|K_{-\infty}^{t},Z_1^\infty)
    \nonumber\\
  &= H(K_{t+1}|K_{-\infty}^{t},Z_1^\infty)+H(Z_{K_{t+1}}|K_{-\infty}^{t+1},Z_1^\infty)
  % \nonumber\\
  % &
    = H(K_{t+1}|K_{-\infty}^{t}),
\end{align}
which is the left-most equality in (\ref{EntropyRateEq}).

Now, we prove the right-most equality in (\ref{EntropyRateEq}).
Applying the calculus of Shannon information measures for
$\sigma$-fields \citep[Sections 5.3 and 5.4]{Debowski21}, we observe
that
\begin{align}
  H(X_t|X_{t-k}^{t-1})&=\mean\okra{-\log P(X_t|X_{t-k}^{t-1})},
  \\
  H(X_t|X_{-\infty}^{t-1})&=\mean\okra{-\log P(X_t|X_{-\infty}^{t-1})}.
\end{align}
We also notice that
$\lim_{k\to\infty} P(X_t|X_{t-k}^{t-1})=P(X_t|X_{-\infty}^{t-1})$
almost surely by the L\'evy law and
\begin{align}
  \mean\sup_{k\in\mathbb{N}}\okra{-\log P(X_t|X_{t-k}^{t-1})}
  \le H(X_t)+\log e<\infty
\end{align}
by a lemma of \citet{Breiman57} and the assumption
$H(X_t)<\infty$. Hence by the Lebesgue dominated convergence, we
obtain
\begin{align}
  \lim_{k\to\infty} H(X_t|X_{t-k}^{t-1})=H(X_t|X_{-\infty}^{t-1}).
\end{align}
It remains to apply the Toeplitz lemma
\begin{align}
  \lim_{k\to\infty}a_k=a \implies
  \lim_{n\to\infty}\frac{1}{n}\sum_{k=1}^{n} a_k=a,
\end{align}
and the Fekete lemma (\ref{Fekete}) to derive the right-most equality
in (\ref{EntropyRateEq}).

\subsection{Proof of Theorem \ref{theoHilbergNeural}}
\label{secHilbergNeural}

We begin the proof with an analysis of a function.  Fix a
$\beta\in(0,1)$.  We denote the function
\begin{align}
  f(s)=f(s,t,c):=(t+s)^{\beta-1}-\frac{c}{s}
\end{align}
for $t,c\ge 0$. Let $r=r(t,c):=\argmax_{s>0} f(s,t,c)$.  We have
\begin{align}
  0
  =
  \left.\frac{df(s)}{ds}\right|_{s=r}
  =
  -(1-\beta)(t+r)^{\beta-2}+\frac{c}{r^2}.
\end{align}
Consequently, $(1-\beta)(t+r)^{\beta-2}r^2 = c$.

Assume first that $t>0$. Then we may bound
\begin{align}
  (1-\beta)t^{\beta-2}r^2
  \ge
  (1-\beta)(t+r)^{\beta-2}r^2
  \ge
  (1-\beta)t^\beta\okra{\frac{r}{t+r}}^2
  .
\end{align}
Hence we obtain the sandwich bound $r_0\le r\le r_2$, where
\begin{align}
  r_0=r_0(t,c)&:=ty,
  &
  r_2=r_2(t,c)&:=\frac{ty}{1-y},
  &
  y&:=\sqrt{\frac{ct^{-\beta}}{1-\beta}}.
\end{align}
Function $f(s)$ is increasing for $s\le r$ and decreasing for
$s\ge r$.  Suppose that $s\le r_2$. For $q=\ceil{r_2/s}$, we may bound
\begin{align}
  r\le sq\le s\okra{\frac{r_2}{s}+1}\le r_2+s\le 2r_2
\end{align}
so we can also bound
\begin{align}
  f(sq,t,c)&\ge
             f(2r_2)
             =t^{\beta-1}\okra{\frac{1-y}{1+y}}^{1-\beta}
             -\frac{c}{t}\okra{\frac{1-y}{2y}}
             .
\end{align}

Now assume that $t=0$ and $c=n$. Then
\begin{align}
  r=r(0,n):=\okra{\frac{n}{1-\beta}}^{\frac{1}{\beta}}
\end{align}
Suppose that $s\le r$. For $q=\ceil{r/s}$, we may bound
\begin{align}
  r\le sq\le s\okra{\frac{r}{s}+1}\le r+s\le 2r
\end{align}
so we can also bound
\begin{align}
  f(sq,0,n)&\ge
         f(2r)
         =\frac{2^{\beta}-1+\beta}{2}
         \okra{\frac{n}{1-\beta}}^{1-\frac{1}{\beta}}
         .
\end{align}
This completes the analysis of function $f(s)$ that will be needed
next.

Now we proceed to the main part of the proof.  By an iterated
application of inequality (\ref{Subadditive}), for any
$q\in\mathbb{N}$, we have inequality
\begin{align}
  \sup_{k\ge t} \frac{H(X_{k+1}^{k+s}|Q)}{s}
  \ge
  \sup_{k\ge t} \frac{H(X_{k+1}^{k+sq}|Q)}{sq}.
\end{align}
By contrast, by inequality (\ref{Triangle}), conditions
(\ref{ComputeBound}) and (\ref{ParameterBound}) imply inequality
\begin{align}
  H(X_{k+1}^{k+s}|Q)
  \ge \max\klam{H(X_{k+1}^{k+s}|X_1^t)-A_H^-c,H(X_{k+1}^{k+s})-A_H^-n}.
\end{align}
Consequently, for $s\le r_2(t,c)$ and $q=\ceil{r_2(t,c)/s}$,
condition (\ref{HilbergConditional}) yields
\begin{align}
  \sup_{k\ge t} \frac{H(X_{k+1}^{k+s}|Q)}{s}-h
  &\ge\sup_{k\ge t} \frac{H(X_{k+1}^{k+sq}|Q)}{sq}-h
    \ge
    \sup_{k\ge t} \frac{H(X_{k+1}^{k+sq}|X_1^t)-A_H^-c}{sq}-h
    \nonumber\\
  &\ge
    A_H^-f(sq,t,c)
    \ge
    A_H^-\okra{t^{\beta-1}\okra{\frac{1-y}{1+y}}^{1-\beta}
    -\frac{c}{t}\okra{\frac{1-y}{2y}}}
    .
\end{align}
Complementing the above inequality with the analogous development for
conditions $s\le r(0,n)$ and $q=\ceil{r(0,n)/s}$ yields
\begin{align}
   \sup_{k\ge t} \frac{H(X_{k+1}^{k+s}|Q)}{s}-h
  &\ge\sup_{k\ge t} \frac{H(X_{k+1}^{k+sq}|Q)}{sq}-h
    \ge
    \sup_{k\ge t} \frac{H(X_{k+1}^{k+sq})-A_H^-n}{sq}-h
    \nonumber\\
  &\ge
    A_H^-f(sq,0,n)
    \ge
    A_H^-\frac{2^{\beta}-1+\beta}{2}
    \okra{\frac{n}{1-\beta}}^{1-\frac{1}{\beta}}
    .
\end{align}
Hence, by the source coding inequality (\ref{SourceCoding}), we
recover the claim (\ref{NeuralConditional}).   

%\printbibliography
%\bibliographystyle{IEEEtran}
\bibliographystyle{abbrvnat}

\setlength{\bibsep}{0pt}

\bibliography{0-journals-abbrv,0-publishers-abbrv,ai,ql,mine,math,tcs,books,nlp}

\end{document}